\documentclass[a4paper,10pt]{article}		
\usepackage[utf8]{inputenc}			
\usepackage[T1]{fontenc}			
\usepackage[parfill]{parskip}			

\usepackage[noblocks]{authblk}

\usepackage{lipsum}

\usepackage{xcolor}
\usepackage{adjustbox}
\usepackage[most]{tcolorbox}

\usepackage{xcolor} 
\usepackage{listings}

\lstdefinestyle{Terminal} {
  backgroundcolor=\color{black},
  basicstyle=\scriptsize\color{white}\ttfamily
}

\usepackage{mathrsfs}

\usepackage{amsthm} 				
\usepackage{amsmath}				
\usepackage{amssymb}				
\usepackage{mathtools}				

\usepackage{graphicx}				
\usepackage{fancyhdr} 				

\usepackage{listings}				
\usepackage{color}					
\usepackage{textcomp}				

\usepackage{calc}

\usepackage{eurosym}				

\usepackage{url}
\urldef\grassmannurl\url{https://proofwiki.org/wiki/Grassmann%27s_Identity}

\usepackage[left=2.00cm, right=2.00cm, top=2.00cm, bottom=2.00cm]{geometry}

\usepackage{ifthen}				

\usepackage{hyperref}

\setlength\headheight{5mm}

\fancyhead[L]{sLORETA zero error proof}
\fancyhead[C]{}
\fancyhead[R]{Malte Höltershinken}

\widowpenalty = 10000
\clubpenalty = 10000
\displaywidowpenalty = 10000

\theoremstyle{definition}
\newtheorem{theorem}{Theorem}
\newtheorem{defin}{Definition}

\newtheorem{lemma}{Lemma}
\newtheorem*{rem}{Remark}

\renewcommand{\phi}{\varphi}						

\DeclareMathOperator{\diag}{diag}
\DeclareMathOperator{\Id}{Id}

\DeclareMathOperator*{\argmin}{arg\,min}
\DeclareMathOperator*{\argmax}{arg\,max}
\DeclareMathOperator{\Image}{Im}
\DeclareMathOperator{\Ker}{Ker}
\DeclareMathOperator{\rank}{rank}

\DeclareMathOperator{\Expected}{\mathbb{E}}
\DeclareMathOperator{\rv}{rv}
\DeclareMathOperator{\Cov}{Cov}
\DeclareMathOperator{\Trace}{Tr}

\newlength{\laengeOben}
\newlength{\laengeUnten}

\newcommand{\stackrelC}[2]{
\settowidth{\laengeOben}{${{}^{#1}}$} \settowidth{\laengeUnten}{#2}
\addtolength{\laengeOben}{- \laengeUnten}
\ifthenelse{\lengthtest{\laengeOben > 0pt}}{\kern -0.5 \laengeOben}{}
\stackrel{#1}{#2}
}

\setcounter{secnumdepth}{5}		

\title{sLORETA is equivalent to single dipole scanning}

\author[1, 2]{Malte B. Höltershinken\footnote{Correspondence: \href{mailto:m_hoel20@uni-muenster.de}{m\_hoel20@uni-muenster.de}}}
\author[1, 2]{Tim Erdbrügger}
\author[1, 3]{Carsten H. Wolters}

\affil[1]{Institute for Biomagnetism and Biosignalanalysis, University of Münster, Münster, Germany}
\affil[2]{Faculty of Mathematics and Computer Science, University of Münster, Münster, Germany}
\affil[3]{Otto Creutzfeldt Center for Cognitive and Behavioral Neuroscience, University of Münster, Münster, Germany}

\date{August 25, 2024}

\begin{document}

\maketitle

\begin{abstract}
In this paper, we show that each member of a family of reconstruction approaches, which includes the classical sLORETA approach and the eLORETA approach, is exactly equivalent to a single dipole scan in some inner product norm. Additionally, we investigate NAI/AG and SAM/UNG beamformers and show that, in scenarios with a single active source with additive uncorrelated noise, they are also equivalent to some form of dipole scanning. As a particular consequence, this shows that, under ideal conditions in a noisy single source case, NAI and SAM (resp. AG and UNG) beamformers are equivalent. Finally, we also discuss how to generalize scalar beamformers to vector beamformers.
\end{abstract}

\section{Introduction}
sLORETA \cite{sloreta_paper} is a widely used approach for source reconstruction. The approach is typically derived as a reweighting of the minimum norm reconstruction and interpreted as a distributed inverse method. One of its appealing properties is that it can reconstruct noiseless dipolar signals \cite{sekihara_neuroimage_loc, greenblatt_ung_unbiased}. We will show that sLORETA and its generalizations defined in \cite{pascualmarqui_eloreta} are single dipole scans in inner product spaces. From this, it becomes immediately clear why sLORETA and its generalizations can reconstruct noiseless dipolar sources. Additionally, this gives an easily understandable meaning to the output generated by sLORETA and thus aids in its interpretation for source reconstruction. We will only present the vector sLORETA case, as the scalar case can be handled using the same strategies with simpler arguments.

Furthermore, if one assumes noisy signals, it was shown that for white noise, sLORETA introduces a bias into the reconstruction \cite{sekihara_neuroimage_loc, greenblatt_ung_unbiased}. In \cite{pascualmarqui_eloreta}, the author responds by showing that if the noise consists of a certain combination of sensor noise and biological noise, sLORETA is unbiased. We will show that these observations are special cases of a more general result. Concretely, we will show that in the presence of noise, a dipole scan is unbiased if, and only if, the inner product is induced by the inverse noise covariance matrix. Using the interpretation of sLORETA as a dipole scan, we will show how to derive the aforementioned results as special cases of this statement.

In the second part of the paper, we will investigate beamforming approaches. More concretely, we will show that if the underlying signal is generated by a single dipole with additive uncorrelated noise, the neural activity index (NAI) based beamforming \cite{van_veen_lcmv_beamforming}, synthetic aperture magnetometry (SAM) beamforming \cite{vrba_sam_lcmv_comparison}, array gain (AG) beamforming and unit noise gain (UNG) beamforming \cite{sekihara_beamforming} approaches can also be interpreted as some form of dipole scans, for NAI and SAM with the inner product induced by the noise covariance matrix, and for AG and UNG with the standard L2 inner product. This generalizes a previous result from \cite{hillebrand_sam_dipole_scanning}, where the authors show that, under the assumption of uncorrelated white noise and a single active source, the scalar SAM beamformer is equivalent to a single dipole scan. A particularly interesting consequence of our results is that under ideal conditions NAI and SAM (resp. AG and UNG) are exactly equivalent for single source scenarios. Finally, we discuss different vector versions of the NAI and show that the original definition in \cite{van_veen_lcmv_beamforming} in general leads to a biased reconstruction, while the version introduced in \cite{moiseev_mcmv_beamforming} is unbiased.

\section{sLORETA and single dipole scanning}
\subsection{Definitions}
We begin by giving all the necessary definitions.

\begin{defin}[Inner product norm]
Let $C$ be a positive definite operator on some inner product space. Then we can define another inner product via
\[
\langle x, y\rangle_C = \langle C \cdot x, y \rangle
\]
and a corresponding norm via
\[
\|x\|_C = \sqrt{\langle C \cdot x, x\rangle}.
\]
Note that the choice $C = \Id$ reproduces the original inner product and norm. Furthermore, note that, for a finite-dimensional space, the set of all inner products is in one-to-one correspondence to the set of positive definite operators.
\end{defin} 

\begin{defin}[Generalized dipole scan]
Let $C$ be a positive definite operator defining an inner product norm $\|\cdot\|_C$. Let $A \in \mathbb{R}^{N \times 3}$ denote a leadfield and $d \in \mathbb{R}^N$ a measurement vector. We then define the \textit{reconstructed dipole moment} to be
\[
j_C(d, A) = \argmin_{j \in \mathbb{R}^3} \|d - A \cdot j\|_C^2.
\]
We then define the \textit{goodness of fit (GOF)} of the leadfield $A$ to be
\[
\text{GOF}_C(d, A) = 1 - \frac{\|d - A \cdot j_C(d, A)\|_C^2}{\|d\|_C^2}.
\]
The \textit{dipole scan} is then given by searching for the leadfield which maximizes the goodness of fit.
\end{defin}

The standard GOF-scan now corresponds to the choice $C = \Id$ on $\mathbb{R}^N$ equipped with the standard inner product.

\begin{defin}[generalized sLORETA, vector case]
\label{sloreta_definition}
Let $C \in \mathbb{R}^{N \times N}$ be a symmetric matrix defining a positive definite operator on a space containing the range of the leadfields. Let $A \in \mathbb{R}^{N \times 3}$ denote a leadfield, and let $d \in \mathbb{R}^N$ denote a measurement vector. We then define the \textit{generalized sLORETA reconstruction} to be
\[
j^{\text{\,sLORETA}}_{C}(d, A) = \left(A^\intercal C A \right)^{-\frac{1}{2}} A^\intercal C d \in \mathbb{R}^3.
\]
As defined in \cite{pascualmarqui_eloreta}, the generalized sLORETA method now consists of searching for large values of $\|j^{\text{\,sLORETA}}_{C}(d, A)\|^2$.
\end{defin}

The classical sLORETA approach now corresponds to the choice $C = \left(L L^\intercal + \alpha H\right)^+$, where $L = (L_1, ..., L_M) \in \mathbb{R}^{N \times (3M)}$ is the complete leadfield, $H$ is the orthogonal projection onto $\{1\}^\perp$, and $\alpha \geq 0$ (as defined in \cite{sloreta_paper}) respectively $C = \left(L L^\intercal + \alpha \Id\right)^{-1}$ (as defined by \cite{sekihara_beamforming}). The eLORETA approach, introduced in \cite{pascualmarqui_eloreta}, corresponds to the choice $C = \left(L W^{-1} L + \alpha H \right)^+$, where $W = \diag(W_1, ..., W_M)$ is a block diagonal matrix given as the solution of the set of equations
\[
W_i = \left( L_i^\intercal \left( L W^{-1} L^\intercal + \alpha H\right)^+ L_i\right)^{\frac{1}{2}}.
\]

\subsection{The equivalence between sLORETA and single dipole scanning}
We will now state and prove that all generalized sLORETA approaches are equivalent to generalized dipole scans.

\begin{theorem}[Equivalence of sLORETA and dipole scanning]
\label{sloreta_dipole_scan_equivalence_theorem}
Let $C$ be a symmetric matrix defining a positive definite operator on a space containing the range of the leadfields. Let $\|\cdot \|_C$ denote the corresponding inner product norm. Then we have
\[
\|j^{\text{\,sLORETA}}_{C}(d, A)\|^2 = \|d\|_C^2 \cdot \text{GOF}_C(d, A).
\]
Hence we see that the sLORETA estimate at a source position is, up to a scalar multiple, given by the goodness of fit of the leadfield at the investigated source position to the data. In particular, we see that visualizing the sLORETA estimate is the same as visualizing a goodness of fit distribution.
\end{theorem}

To prove this theorem, we will first prove two lemmas.

\begin{lemma}
\label{generalized_least_squares}
Let $C$, $A$, and $d$ be as above. We then have
\[
\argmin_{j \in \mathbb{R}^3} \|d - A \cdot j\|_C^2 = \left( A^\intercal C A \right)^{-1} A^\intercal C d.
\]
\end{lemma}

\begin{proof}
Using $\|x\|_C^2 = \langle C x, x\rangle$, we have for $j \in \mathbb{R}^3$ that
\begin{align*}
\|d - A \cdot j\|_C^2 = \langle C \left(d - A \cdot j \right), d - A \cdot j\rangle
= \langle C^{\frac{1}{2}} \left(d - A \cdot j \right), C^{\frac{1}{2}} \left(d - A \cdot j\right)\rangle
= \|C^{\frac{1}{2}} d - C^{\frac{1}{2}} A \cdot j\|^2.
\end{align*}
Hence we have that
\[
\argmin_{j \in \mathbb{R}^3} \|d - A \cdot j\|_C^2
= \argmin_{j \in \mathbb{R}^3} \|C^{\frac{1}{2}} d - C^{\frac{1}{2}} A \cdot j\|^2,
\]
which is an ordinary least squares problem, whose solution $j_C(d, A)$ is given by the Gauss normal equation
\[
\left( C^{\frac{1}{2}} A \right)^\intercal C^{\frac{1}{2}} A \cdot j_C(d, A) = \left( C^{\frac{1}{2}} A \right)^\intercal C^{\frac{1}{2}} d
\iff j_C(d, A) = \left(A^\intercal C A \right)^{-1} A^\intercal C d.
\]
\end{proof}

\begin{lemma}
\label{gof_expression}
Let $C$, $A$, and $d$ be as above. We then have
\[
\text{GOF}_C(d, A) = \frac{d^\intercal C A\left(A^\intercal C A \right)^{-1} A^\intercal C d}{\|d\|_C^2}.
\]
\end{lemma}

\begin{proof}
Let $j_C(d, A)$ be defined as in Lemma \ref{generalized_least_squares}. We then have
\begin{align*}
\|d - A \cdot j_C(d, A)\|_C^2 
&= \|d\|_C^2 - 2 \langle d, A \cdot j_C(d, A) \rangle_C + \|A \cdot j_C(d, A)\|_C^2 \\
&= \|d\|_C^2 - 2 d^\intercal C A \left(A^\intercal C A \right)^{-1} A^\intercal C d + \|A \cdot j_C(d A)\|_C^2.
\end{align*}
We can now compute
\begin{align*}
\|A \cdot j_C(d, A)\|_C^2 &= j_C(d, A)^\intercal A^\intercal C A j_C(d, A) \\
&= d^\intercal C A \left(A^\intercal C A \right)^{-1} A^\intercal C A \left(A^\intercal C A \right)^{-1} A^\intercal C d\\
&= d^\intercal C A \left(A^\intercal C A \right)^{-1} A^\intercal C d,
\end{align*}
and thus
\[
\|d - A \cdot j_C(d, A)\|_C^2 = \|d\|_C^2 - d^\intercal C A \left(A^\intercal C A \right)^{-1} A^\intercal C d.
\]
We thus have
\[
\text{GOF}_C(d, A) = 1 - \frac{\|d - A \cdot j_C(d, A)\|_C^2}{\|d\|_C^2}
= \frac{d^\intercal C A \left(A^\intercal C A \right)^{-1} A^\intercal C d}{\|d\|_C^2}.
\]
\end{proof}

\begin{proof}[Proof of theorem \ref{sloreta_dipole_scan_equivalence_theorem}]
By definition \ref{sloreta_definition} and lemma \ref{gof_expression}, we have
\begin{align*}
\|j^{\text{\,sLORETA}}_{C}(d, A)\|^2
= \| \left(A^\intercal C A \right)^{-\frac{1}{2}} A^\intercal C d \|^2
= d^\intercal C A \left(A^\intercal C A \right)^{-1} A^\intercal C d
= \|d\|_C^2 \cdot \text{GOF}_C(d, A).
\end{align*}
\end{proof}

\subsection{The impact of noise}
\label{noise_impact_section}
We now want to investigate a \textbf{noisy single source setting}. Let $L_0 \in \mathbb{R}^{N \times 3}$ be the leadfield corresponding to some point in the head, and let $0 \neq \eta \in \mathbb{R}^3$ be some orientation. We now assume our signal $d(t)$ is of the form
\[
d(t) = L_0 \cdot (s(t) \cdot \eta) + n(t),
\]
where $s(t)$ is the (real) signal amplitude and $n(t)$ is the noise vector. We assume that $s(t)$ is uncorrelated to the components of $n(t)$. We furthermore assume that $Q^2 := \Expected [ s(t)^2] > 0$, $\Expected [n(t)] = 0$, and that $s(t)$ and the components of $n(t)$ are square integrable. Let $N = \Cov(n(t))$ denote the noise covariance matrix.

In the following, we assume that $N$ is invertible. We want to emphasize that this assumption is only used to keep the exposition simple. When restricting to the orthogonal complement of the kernel of $N$ and adapting the statements accordingly, all results remain true. 

Let $A \in \mathbb{R}^{N \times 3}$ be the leadfield at some position we are currently scanning. Given some positive definite operator $C$, we then define the \textbf{residual variance} in the corresponding norm to be
\[
\rv_{C}(d(t), A) = \min_{j \in \mathbb{R}^3} \|d(t) - A \cdot j\|_C^2.
\]
We now want to investigate under what conditions we can guarantee that the expected residual variance is minimal for the true leadfield, i.e. when is the statement
\[
L_0 = \argmin_{\text{$A$ leadfield}} \, \Expected \left[\, \rv_C(d(t), A)\,\right]
\]
true? The answer to this question is given by the following theorem.

\begin{theorem}
\label{minimizer_in_the_presence_of_noise_theorem}
Without further assumptions on the leadfields, we can guarantee
\begin{equation}
\label{minimizer_in_the_presence_of_noise_equation}
L_0 = \argmin_{\text{$A$ leadfield}} \, \Expected \left[\, \rv_C(d(t), A)\,\right]
\end{equation}
if, and only if, $C = \alpha \cdot N^{-1}$ for some $\alpha > 0$. 

\bigskip
\end{theorem}

Note that the choice $C = N^{-1}$ corresponds to a pre-whitening of the data. This theorem will be the consequence of a set of lemmas we will now prove.

\begin{lemma}
\label{minimum_rv_characterization}
We have that
\[
\argmin_{\text{$A$ leadfield}} \, \Expected \left[\, \rv_C(d(t), A)\,\right]
= \argmax_{\text{$A$ leadfield}} \,\, Q^2 \eta^\intercal \left(L_0^\intercal C A \right) \left(A^\intercal C A \right)^{-1} \left(A^\intercal C L_0 \right) \eta + \Trace \left( \left(A^\intercal C A \right)^{-1} A^\intercal C N C A\right).
\]
\end{lemma}

\begin{proof}
Using theorem \ref{sloreta_dipole_scan_equivalence_theorem}, we can compute that
\begin{align*}
\argmin_{\text{$A$ leadfield}} \, \Expected \left[\, \rv_C(d(t), A)\,\right] 
&= \argmax_{\text{$A$ leadfield}} \, \left(\Expected \left[ \,\|d(t)\|_C^2 \,\right] -  \Expected \left[\, \rv_C(d(t), A)\,\right]\right)
= \argmax_{\text{$A$ leadfield}} \, \Expected \left[ \, \|d(t)\|_C^2 -  \rv_C(d(t), A)\,\right] \\
&= \argmax_{\text{$A$ leadfield}} \, \Expected \left[ \, \|d\|_C^2 \cdot \text{GOF}_C(d(t), A)\,\right]
\stackrel{\text{Theorem \ref{sloreta_dipole_scan_equivalence_theorem}}}{=}
\argmax_{\text{$A$ leadfield}} \,\Expected \left[ \|j^{\text{\,sLORETA}}_{C}(d(t), A)\|^2 \right],
\end{align*}
and since
\begin{multline*}
\Expected \left[ \|j^{\text{\,sLORETA}}_{C}(d(t), A)\|^2 \right]
= \Expected \left[\, d(t)^\intercal C A \left(A^\intercal C A \right)^{-1} A^\intercal C d(t) \, \right] \\
= \Expected [\, 
s(t)^2 \eta^\intercal \left(L_0^\intercal C A \right) \left(A^\intercal C A \right)^{-1} \left(A^\intercal C L_0 \right) \eta
+ n(t)^\intercal C A \left(A^\intercal C A \right)^{-1} A^\intercal C n(t) \\
+ \left(s(t) n(t)^\intercal\right) C A \left(A^\intercal C A \right)^{-1} A^\intercal C L_0 \eta
+ \eta^\intercal L_0^\intercal C A \left(A^\intercal C A \right)^{-1} A^\intercal C \left(s(t) n(t)\right)] \\
= Q^2 \eta^\intercal \left(L_0^\intercal C A \right) \left(A^\intercal C A \right)^{-1} \left(A^\intercal C L_0 \right) \eta
+ \Trace\left(C A \left(A^\intercal C A \right)^{-1} A^\intercal C N\right),
\end{multline*}
the statement of the lemma follows.
\end{proof}

Now let $d_0 = Q L_0 \eta$. Then
\[
Q^2 \eta^\intercal \left(L_0^\intercal C A \right) \left(A^\intercal C A \right)^{-1} \left(A^\intercal C L_0 \right) \eta
= d_0^\intercal C A \left(A^\intercal C A \right)^{-1} A^\intercal C d_0
\stackrel{\text{Lemma \ref{gof_expression}}}{=} \|d_0\|_C^2 \cdot \text{GOF}_C(d_0, A),
\]
and we see that the right-hand side of lemma \ref{minimum_rv_characterization} is essentially given by the sum of the goodness of fit of the leadfield $A$ to the noiseless dipolar signal $d_0 = Q L_0 \eta$ and a noise-dependent distortion.
We can now prove one half of theorem \ref{minimizer_in_the_presence_of_noise_theorem}.

\begin{proof}[Proof that $C = \alpha \cdot N^{-1}$ is sufficient for equation (\ref{minimizer_in_the_presence_of_noise_equation}) to hold.]
If $C = \alpha \cdot N^{-1}$, we have
\[
\Trace\left( \left( A^\intercal C A\right)^{-1} A^\intercal C N C A\right) = 3 \alpha,
\]
and hence
\[
\argmin_{\text{$A$ leadfield}} \, \Expected \left[\, \rv_C(d(t), A)\,\right]
= \argmax_{\text{$A$ leadfield}}  \, \|d_0\|_C^2 \cdot \text{GOF}_C(d_0, A)
= L_0.
\]
\end{proof}

Before we start to approach the other half of theorem \ref{minimizer_in_the_presence_of_noise_theorem}, we first introduce some notation. 

\begin{defin}
Let $V, W$ denote normed vector spaces, let $U \subset V$ be open, and let $f : U \rightarrow V$ be a map. We then say that $F$ is \textit{Frechet differentiable}, or simply \textit{differentiable}, in $x \in U$ if there exists a bounded linear operator $T : V \rightarrow W$ such that
\[
f(x + h) = f(x) + T(h) + \varphi(h),
\]
where $\varphi (h)\in o(h)$. Such an operator $T$, if it exists, is unique. We then write $T = Df(x)$ and call $Df(x)$ the \textit{derivative} of $f$ at $x$.
\end{defin}

The usual properties of derivatives, to the extent that they make sense in general normed spaces, carry over to Frechet derivatives, see e.g. \cite{dieudonne_foundations_of_modern_analysis}. One can for example show that the chain rule holds. Furthermore, a scalar differentiable function can only have a local optimum in $x$ if $Df(x) = 0$. In particular, if $Df(x) \neq 0$ then $f$ can not have a local maximum at $x$.

Our strategy for the second half of theorem \ref{minimizer_in_the_presence_of_noise_theorem} is now as follows. Let $O = \{A \in \mathbb{R}^{N \times 3} | \rank (A) = 3\}$. Then $O \subset \mathbb{R}^{N \times 3}$ is open and dense, and the maps
\begin{equation}
f : O \rightarrow \mathbb{R}; A \mapsto Q^2 \eta^\intercal \left(L_0^\intercal C A \right) \left(A^\intercal C A \right)^{-1} \left(A^\intercal C L_0 \right) \eta
\end{equation}
and
\begin{equation}
\label{noise_functional_definition}
g : O \rightarrow \mathbb{R}; A \mapsto  \Trace \left( \left(A^\intercal C A \right)^{-1} A^\intercal C N C A\right)
\end{equation}
are well defined. We will now investigate under what conditions we have $D(f + g)(L_0) \neq 0$, which by lemma \ref{minimum_rv_characterization} will lead to a proof of theorem \ref{minimizer_in_the_presence_of_noise_theorem}.

First, note that as derived above we have $f(A) = \|d_0\|_C^2 \cdot \text{GOF}_C(d_0, A)$, and hence $f(A)$ has a local maximum at $L_0$. Thus we must have $Df(L_0) = 0$. Furthermore, a straightforward computation reveals that for $A, B \in \mathbb{R}^{N \times 3}$ we have
\begin{equation}
\label{derivative_noise_distortion}
Dg(A)(B) = 2 \cdot \Trace \left( B \cdot \left( \left(A^\intercal C A\right)^{-1} A^\intercal C N C - \left(A^\intercal C A\right)^{-1} A^\intercal C N C A \left(A^\intercal C A\right)^{-1} A^\intercal C \right)\right).
\end{equation}

In the appendix, we show the calculation of this derivative.

\begin{lemma}
The following are equivalent.
\begin{enumerate}
\item $Dg(A) = 0$
\item $C \cdot N \cdot \Ker(A^\intercal) = \Ker(A^\intercal)$
\end{enumerate}
\end{lemma}

\begin{proof}
Assume $Dg(A) = 0$.
First note that $\Trace(B \cdot M) = 0$ for all $B \in \mathbb{R}^{N \times 3}$ implies $M = 0$, as can be seen from taking $B = M^\intercal$ and remembering that the trace of a positive semidefinite matrix is the sum of its non-negative eigenvalues. Hence, equation (\ref{derivative_noise_distortion}) implies
\[
A^\intercal C N \left(E_N - C A \left(A^\intercal C A \right)^{-1} A^\intercal \right) = 0
\]
Hence, for $v \in \Ker(A^\intercal)$ we have $A^\intercal C N v = 0$, and hence $C \cdot N \cdot \Ker(A^\intercal) \subset \Ker(A^\intercal)$. Since $\dim(C \cdot N \cdot \Ker(A^\intercal)) = \dim(\Ker(A^\intercal))$, we must in fact have equality.

On the other hand, assume that we have $C \cdot N \cdot \Ker(A^\intercal) = \Ker(A^\intercal)$. Then for $v \in \Ker(A^\intercal)$ we have
\[
A^\intercal C N \left(E_N - C A \left(A^\intercal C A \right)^{-1} A^\intercal \right) \cdot v = 0.
\]
If we now have $v \in \Image(C A)$, we can write $v = C A z$. But then
\[
A^\intercal C N \left(E_N - C A \left(A^\intercal C A \right)^{-1} A^\intercal \right) \cdot v
= A^\intercal C N \left(C A z - C A z\right) = 0.
\]
Since $\mathbb{R}^N = \Ker(A^\intercal) \bigoplus \Image(C A)$, this implies
\[
A^\intercal C N \left(E_N - C A \left(A^\intercal C A \right)^{-1} A^\intercal \right) = 0,
\]
which by equation (\ref{derivative_noise_distortion}) implies $Dg(A) = 0$.
\end{proof}

\begin{proof}[Proof of the second half of theorem \ref{minimizer_in_the_presence_of_noise_theorem}]
In particular, we thus see that $D(f + g)(L_0) = 0$ if and only if 
\[
C \cdot N \cdot \Ker(L_0^\intercal) = \Ker(L_0^\intercal),
\]
 i.e. if the kernel of $L_0^\intercal$ is an invariant subspace of $C \cdot N$. If $C \neq \alpha \cdot N^{-1}$ for all $\alpha > 0$, and we make no further assumption on $L_0$, we can not guarantee that $C \cdot N \cdot \Ker(L_0^\intercal) = \Ker(L_0^\intercal)$, and hence we cannot guarantee that $D(f + g)(L_0) = 0$. But if $D(f + g)(L_0) \neq 0$, there is some $\widetilde{L} \neq L$ such that $(f + g)(\widetilde{L}) > (f + g)(L_0)$, and hence
 \[
 L_0 \neq \argmin_{\text{$A$ leadfield}} \, \Expected \left[\, \rv_C(d(t), A)\,\right].
 \]
\end{proof}

\begin{rem}
In the proof of theorem \ref{minimum_rv_characterization}, we derived the equality
\[
\argmin_{\text{$A$ leadfield}} \, \Expected \left[\, \rv_C(d(t), A)\,\right]  = \argmax_{\text{$A$ leadfield}} \,\Expected \left[ \|j^{\text{\,sLORETA}}_{C}(d(t), A)\|^2 \right].
\]
When the authors in \cite{sekihara_neuroimage_loc, greenblatt_ung_unbiased, pascualmarqui_eloreta} say sLORETA is either biased or unbiased, what they mean is whether the equation
\[
L_0 =  \argmax_{\text{$A$ leadfield}} \,\Expected \left[ \|j^{\text{\,sLORETA}}_{C}(d(t), A)\|^2 \right]
\]
holds or not. We thus see that the corresponding results in these papers are indeed special cases of theorem \ref{minimizer_in_the_presence_of_noise_theorem}.
\end{rem}

\section{Beamforming and dipole scanning}

\subsection{Introduction to beamforming}

We have seen above that any member of the generalized sLORETA family of reconstruction approaches described in \cite{pascualmarqui_eloreta} can be described as a single dipole scan. In particular, this gives an intuitive and rigorous explanation of why sLORETA can reconstruct single sources in noiseless scenarios. 

In particular, even though sLORETA was derived using a statistical argument, it turned out that its success could be traced back to a least squares fit. With this in mind, it is natural to ask if other reconstruction approaches can also be related to least squares fits in a similar manner.

One famous class of reconstruction approaches are so-called \textit{beamformers}, or \textit{adaptive spatial filters} \cite{sekihara_beamforming}. In EEG and MEG, the underlying idea of these approaches is that, given some dipolar source of interest, one wants to create a linear functional which, when applied to the data, suppresses contributions that do not originate from the dipole, while letting contributions from this dipole pass through. In the following, we want to focus on the class of \textit{minimum variance beamformers}.

\begin{defin}[Scalar beamformer]
Let $l \in \mathbb{R}^{n}$ be the leadfield vector from the dipole of interest. Let $d$ be the data and $R = \mathbb{E}[d d^\intercal]$ its covariance matrix\footnote{Note that $R$ is, strictly speaking, the second moment matrix of the signal $d$. But as is common in the beamforming literature, we have followed the convention of \cite{sekihara_beamforming} and used the term covariance matrix.}. Furthermore, let $\tau > 0$ be a parameter. Then a scalar beamformer $w_{\tau}(l)$ is defined by
\[
w_{\tau}(l) = \argmin_{w^\intercal l = \tau} \, w^\intercal R w.
\]
\end{defin}
Using Lagrange multipliers, one can compute that 
\begin{equation}
\label{filter_explicit_formula}
w_{\tau}(l) = \tau \cdot \frac{R^{-1} l}{l^\intercal R^{-1} l}.
\end{equation}

The motivation behind this definition is that, while the side constraint ensures that the signal of interest is passed, the minimization of $w^\intercal R w = \mathbb{E}[(w^\intercal d)^2]$ forces the filter to adaptively learn what other contributions are present in the data, and how these contributions can be suppressed.

Different beamformers now correspond to different choices of $\tau$. In particular, for a fixed $l$ the beamformers only differ by a multiplicative constant. If we thus have data samples $(d_1, ..., d_T)$ and apply a beamformer to get an estimation of the time course of the dipolar source with leadfield $l$, the reconstructed timecourses $(w_\tau(l)^\intercal d_1, ..., w_\tau(l)^\intercal d_T)$ for different beamformers also only differ by a multiplicative constant. 

The real benefit of introducing the parameter $\tau$ comes when comparing the beamformer output for different dipolar sources. Assume you want to reconstruct the position and orientation of a dipolar source that was active during the measurement. The common heuristic in beamforming is to estimate the source position-orientation pair $(\hat{x}_0, \hat{\eta}_0)$ by scanning the source space and maximizing the expected power of the post-filter signal, or more concretely
\[
(\hat{x}_0, \hat{\eta}_0) = \argmax_{x, \eta} \, \mathbb{E}[\left(w_\tau(l(x, \eta))^\intercal d\right)^2],
\]
where $l(x, \eta) \in \mathbb{R}^N$ denotes the leadfield vector for a source at position $x$ with orientation $\eta$.

Using (\ref{filter_explicit_formula}), we can compute
\[
P_\tau(l) := \mathbb{E}[\left(w_\tau(l(x, \eta))^\intercal d\right)^2] = \frac{\tau^2}{l^\intercal R^{-1} l}.
\]

We now want to introduce the common strategies for choosing $\tau$.

\begin{defin}[UG, AG, NAI, UNG, SAM]
Let $l \in \mathbb{R}^N$ denote an arbitrary leadfield vector. Furthermore, let $N$ denote the noise covariance matrix.
\begin{enumerate}
\item The \textit{Unit Gain (UG)} beamformer corresponds to the choice $\tau = 1$. Its expected output power is given by
\[
P_{\text{UG}}(l) := P_1(l) = \frac{1}{l^\intercal R^{-1} l}
\] 
\item The \textit{Neural Activity Index (NAI)} beamformer corresponds to the choice $\tau = \|N^{-\frac{1}{2}} l\|_2$. Its expected output power is given by
\begin{equation}
\label{scalar_nai}
P_{\text{NAI}}(l) := P_{\|N^{-\frac{1}{2}} l\|_2}(l) = \frac{l^\intercal N^{-1} l}{l^\intercal R^{-1} l}.
\end{equation}
In the special case $N = \sigma^2 \Id$, this choice is also called \textit{Array Gain (AG)} beamformer.
\item The \textit{Synthetic Aperture Magnetometry (SAM)} beamformer corresponds to the choice $\tau = \frac{l^\intercal R^{-1} l}{\sqrt{l^\intercal R^{-1} N R^{-1} l}}$. Its expected output power is given by
\begin{equation}
\label{scalar_sam}
P_{\text{SAM}}(l) := P_{\frac{l^\intercal R^{-1} l}{\sqrt{l^\intercal R^{-1} N R^{-1} l}}}(l) = \frac{l^\intercal R^{-1} l}{l^\intercal R^{-1} N R^{-1} l}.
\end{equation}
In the special case $N = \sigma^2 \Id$, this choice is also called \textit{Unit Noise Gain (UNG)} beamformer.
\end{enumerate}
\end{defin}

We want to make a few comments on this definition.

\begin{rem}
\begin{enumerate}
\item In the literature, the SAM beamformer is typically defined in an equivalent, but slightly different form. A straightforward computation using equation (\ref{filter_explicit_formula}) shows that for an arbitrary $\tau > 0$ we have
\[
\frac{w_\tau(l)^\intercal R w_\tau(l)}{w_\tau(l)^\intercal N w_\tau(l)} = P_{\text{SAM}}(l).
\]
The left-hand side of this equation is commonly called \textit{pseudo-Z} score and is used to define the SAM beamformer. It can be motivated as ``projected signal power divided by projected noise power'', and SAM beamforming can thus be seen as maximizing some measure of SNR.
\item The value $P_{\text{NAI}}(l)$ is typically called the neural activity index, hence the name of the beamformer. Note that $P_{\text{NAI}}$ can also be interpreted as a measure of SNR.
\item In the definition above, we have phrased the AG and UNG beamformers as special cases of the NAI and SAM beamformers. In particular, it follows that results shown to hold for NAI and SAM beamformers for arbitrary noise covariance matrices, are also valid for AG and UNG beamformers under the assumption of white noise, i.e. for noise covariance matrices of the form $N = \sigma^2 \Id$. But we could just as well have started with defining AG and UNG beamformers for white noise, and then defined NAI beamforming to be AG beamforming applied to pre-whitened data, and SAM beamforming to be UNG beamforming applied to pre-whitened data. Using this, it follows that results shown to hold for AG and UNG under the assumption of white noise directly carry over to results for NAI and SAM for arbitrary noise covariance matrices. From a theoretical perspective, it thus makes no difference if one investigates AG and UNG beamformers under the assumption of white noise, or NAI and SAM beamformers for arbitrary noise covariance matrices.
\end{enumerate}
\end{rem}

Note that the UG beamformer is biased and can reconstruct neither the position nor the orientation of a dipole \cite{sekihara_beamforming, buschermoehle_beamforming}, and we will hence not study it further here. On the other hand, both NAI as well as SAM beamforming can be shown to produce an unbiased estimation of single dipolar sources \cite{moiseev_mcmv_beamforming}. They will be the focus of the following.

The scalar beamformers described above are tailored to a single active source with a fixed orientation. A natural question is then how this can be generalized to multiple active sources, with possibly non-fixed orientation. One such generalization is given in \cite{moiseev_mcmv_beamforming}, and goes as follows.

\begin{defin}[NAI and SAM Vector beamformer]
Let $N$ be the number of sensors. Let $L \in \mathbb{R}^{N \times k}$ denote a matrix whose columns contain the leadfield vectors of suspected active dipoles. Let $R$ and $N$ again denote the signal and noise covariance matrices. Then the vector version of the NAI beamformer power is defined as
\begin{equation}
\label{vector_nai}
P_{\text{NAI}}(L) = \Trace \left( \left(L^\intercal N^{-1} L\right) \left(L^\intercal R^{-1} L \right)^{-1} \right)
\end{equation}
and the vector version of the SAM beamformer power is defined as
\begin{equation}
\label{vector_sam}
P_{\text{SAM}}(L) = \Trace \left( \left( L^\intercal R^{-1} L\right) \left(L^\intercal R^{-1} N R^{-1} L \right)^{-1}\right).
\end{equation}
Note that for $k = 1$ these equations indeed reduce to (\ref{scalar_nai}) and (\ref{scalar_sam}).
\end{defin}

In \cite{moiseev_mcmv_beamforming}, the authors prove that searching for sets of leadfield matrices $L$ such that $P_\text{NAI}(L) - k$ (resp. $P_{\text{SAM}}(L) - k$) is maximal provides an unbiased reconstruction.\footnote{Note that $k$, i.e. the number of columns of $L$, is only relevant for finding maximizers if one includes matrices of different sizes in the search procedure.} Concretely, they show that both of these expressions attain their maximum value if the range of the matrix $L$ contains the span of the leadfield vectors of the active sources.

We want to note that there is another contender for the vector generalization of the NAI. In \cite{van_veen_lcmv_beamforming}, the authors introduce the value
\begin{equation}
\label{nai_van_veen_definition}
\widetilde{\text{NAI}}(L) = \frac{\Trace\left( \left(L^\intercal R^{-1} L\right)^{-1}\right)}{\Trace \left(\left(L^\intercal N^{-1} L\right)^{-1} \right)}
\end{equation}
as neural activity index. In theorem \ref{original_nai_biased_theorem}, we will show that this definition in general leads to a bias in the reconstruction. Hence, we argue that (\ref{vector_nai}) is a more suitable vector version of the neural activity index.

\subsection{Connection to dipole scanning}
In the following, we will again assume that the data $d$ can be described using a noisy single source setting
\[
d =  s \cdot l_0 + n,
\]
where $s$ is the source amplitude, $l_0 \in \mathbb{R}^N$ is the true leadfield vector, and $n$ is additive noise. We again assume that $s$ is uncorrelated to the components of $n$, that $s$ and $n$ are square integrable, and that $n$ is zero mean. A straightforward computation then shows that 
\begin{equation}
R := \mathbb{E}\left[d d^\intercal \right] =  N + x x^\intercal,
\end{equation}
where $N = \mathbb{E}\left[n n^\intercal \right]$ is the noise covariance matrix and $x = \sqrt{\mathbb{E}\left[s^2\right]} \cdot l_0$. 

Notice that, if $R$ and $N$ are known, one can reconstruct the leadfield vector $l_0$, up to a scaling factor. For example, we have
\[
N^{-\frac{1}{2}} R N^{-\frac{1}{2}} = \Id + \left( N^{-\frac{1}{2}} x\right) \cdot \left( N^{-\frac{1}{2}} x\right)^\intercal,
\]
and hence $\mathbb{R} \cdot N^{-\frac{1}{2}} x$ is the eigenspace of $N^{-\frac{1}{2}} R N^{-\frac{1}{2}}$ corresponding to the largest eigenvalue. If we thus know $R$ and $N$, we also know $\mathbb{R} l_0$, and hence also $\text{GOF}_C(l_0, A)$ for an arbitrary positive definite matrix $C$ and leadfield matrix $A \in \mathbb{R}^{n \times k}$.\footnote{Note that $\text{GOF}_C(\lambda \cdot l_0, A) = \text{GOF}_C(l_0, A)$ for all $\lambda \neq 0$, as follows e.g. from lemma \ref{gof_expression}. }

A natural approach for source reconstruction is now given as follows.
\begin{enumerate}
\item Given $R$ and $N$, compute an arbitrary vector $0 \neq v \in \mathbb{R} \cdot l_0$.
\item Using this vector, compute $\text{GOF}_C(l_0, A)$ for a set of candidate leadfields $A$. 
\item Search for positions where this value is large.
\end{enumerate}

It is obvious that this provides an unbiased reconstruction of the underlying source, and the maximum value of $1$ is attained if and only if $l_0$ is contained in the range of the candidate leadfield $A$. Furthermore, note that this procedure is essentially a dipole scan with the data vector derived from the covariance matrices $R$ and $N$.

Our claim is now that SAM and NAI beamforming is equivalent to this dipole scanning approach, both in the scalar and the vector case. As a corollary, this also implies that SAM and NAI beamforming are equivalent in a noisy single source setting.

The formal statement goes as follows.

\begin{theorem}
\label{sam_and_nai_dipole_scanning}
There are strictly monotonically increasing functions $f_{x, N}, g_{x, N} : [0, 1] \rightarrow R$ such that for all full rank $L \in \mathbb{R}^{n \times k}$ we have
\begin{equation}
P_{\text{NAI}}(L) = k + f_{x, N}\left(\text{GOF}_{N^{-1}}(l_0, L)\right)
\end{equation}
and
\begin{equation}
P_{\text{SAM}}(L) = k + g_{x, N}\left(\text{GOF}_{N^{-1}}(l_0, L)\right)
\end{equation}
Concretely, these functions are given by
\[
f_{x, N}(t) = \frac{\mu_{x, N} \cdot t}{1 - \mu_{x, N} \cdot t}
\]
and 
\[
g_{x, N}(t) = \frac{1}{\langle N^{-1} x, x\rangle + 2} \cdot \left(\frac{1}{1 - \rho_{x, N} \cdot t}  - 1\right),
\]
where the constants are given by
\[
\mu_{x, N} = \frac{\langle N^{-1} x, x\rangle}{1 + \langle N^{-1} x, x\rangle}
\]
and
\[
\rho_{x, N} = \frac{\langle N^{-1} x, x\rangle^2 + 2 \langle N^{-1} x, x\rangle}{\left(\langle N^{-1} x, x\rangle + 1\right)^2}.
\]
Note that $0 < \mu_{x, N}, \rho_{x, N} < 1$.
\end{theorem}

\begin{rem}
Note that this immediately implies that in a noisy single source setting, both NAI as-well as SAM produce an unbiased estimation of the source, i.e. they are maximized by leadfields $L$ whose range contains the true leadfield vector. Furthermore, this shows that there is a monotonically increasing function $h$ such that $P_{\text{SAM}}(L) = h(P_{\text{NAI}}(L))$, and we see that the SAM pseudo-Z score and the neural activity index can be transformed into another via using a bijective transformation. They thus contain the same amount of information.
\end{rem}

We will now work towards a proof of theorem \ref{sam_and_nai_dipole_scanning}. During the following computations, we will make use of the equation from lemma \ref{gof_expression}
\[
\|d\|_C^2 \cdot \text{GOF}_C(d, A) = d^\intercal C A\left(A^\intercal C A \right)^{-1} A^\intercal C d
\]
multiple times.
\begin{lemma}
\label{r_inv_lemma}
We use the definitions from theorem (\ref{sam_and_nai_dipole_scanning}). Then
\begin{align*}
R^{-1} N R^{-1} &= N^{-1} - \frac{\rho_{x, N}}{\langle N^{-1} x, x\rangle} \left(N^{-1} x\right)  \left(N^{-1} x\right)^\intercal \\
R^{-1} &= N^{-1} - \frac{\mu_{x, N}}{\langle N^{-1} x, x\rangle}  \left(N^{-1} x\right)  \left(N^{-1} x\right)^\intercal
\end{align*}
\end{lemma}
\begin{proof}
Since $R = N + x x^\intercal$, this follows from a straightforward computation using the Sherman-Morrison formula\footnote{\url{https://en.wikipedia.org/wiki/Sherman\%E2\%80\%93Morrison_formula}}.
\end{proof}

\begin{proof}[Proof of Theorem \ref{sam_and_nai_dipole_scanning}]
We start with the neural activity index. Using lemma \ref{r_inv_lemma}, we have for $L \in \mathbb{R}^{n \times k}$
\begin{equation}
\label{reduced_R}
L^\intercal R^{-1} L = L^\intercal N^{-1} L - \frac{\mu_{x, N}}{\langle N^{-1} x, x\rangle} \left( L^\intercal N^{-1} x\right)  \left( L^\intercal N^{-1} x\right)^\intercal,
\end{equation}
and thus
\[
L^\intercal N^{-1} L = L^\intercal R^{-1} L + \frac{\mu_{x, N}}{\langle N^{-1} x, x\rangle} \left( L^\intercal N^{-1} x\right)  \left( L^\intercal N^{-1} x\right)^\intercal.
\]
This implies
\[
\left( L^\intercal N^{-1} L\right) \cdot \left( L^\intercal R^{-1} L\right)^{-1}
= \Id_k + \frac{\mu_{x, N}}{\langle N^{-1} x, x\rangle} \left( L^\intercal N^{-1} x\right)  \left( L^\intercal N^{-1} x\right)^\intercal \cdot \left( L^\intercal R^{-1} L\right)^{-1}.
\]
Using $\Trace(A B) = \Trace(B A)$, we then have
\[
P_{\text{NAI}}(L) = \Trace\left( \left( L^\intercal N^{-1} L\right) \cdot \left( L^\intercal R^{-1} L\right)^{-1} \right)
= k +  \frac{\mu_{x, N}}{\langle N^{-1} x, x\rangle} \cdot  \left( L^\intercal N^{-1} x\right)^\intercal \cdot \left( L^\intercal R^{-1} L\right)^{-1} \cdot L^\intercal N^{-1} x.
\]
Applying the Sherman-Morrison formula to (\ref{reduced_R}), we have
\[
\left( L^\intercal R^{-1} L\right)^{-1} = \left( L^\intercal N^{-1} L\right)^{-1} + \frac{\mu_{x, N}}{\langle N^{-1} x, x\rangle}  
\cdot \frac{\left( L^\intercal N^{-1} L\right)^{-1} L^\intercal N^{-1} x \left(L^\intercal N^{-1} x \right)^\intercal \left( L^\intercal N^{-1} L\right)^{-1}}{1 - \mu_{x, N} \cdot \text{GOF}_{N^{-1}}(x, L) }
\]
Putting the last two equations together, we get
\begin{align*}
P_{\text{NAI}}(L) - k &= \mu_{x, N} \cdot \text{GOF}_{N^{-1}}(x, L) + \left(\frac{\mu_{x, N}}{\langle N^{-1} x, x\rangle}\right)^2 \cdot \frac{\left(\langle N^{-1} x, x\rangle  \text{GOF}_{N^{-1}}(x, L)\right)^2}{1 - \mu_{x, N} \cdot  \text{GOF}_{N^{-1}}(x, L)}) \\
&= \mu_{x, N} \cdot  \text{GOF}_{N^{-1}}(x, L) \cdot \left(1 + \frac{ \mu_{x, N} \cdot \text{GOF}_{N^{-1}}(x, L)}{1 - \mu_{x, N} \cdot  \text{GOF}_{N^{-1}}(x, L)} \right) \\
&= \frac{\mu_{x, N} \cdot  \text{GOF}_{N^{-1}}(x, L)}{1 - \mu_{x, n} \cdot  \text{GOF}_{N^{-1}}(x, L)}
= f_{x, N}(\text{GOF}_{N^{-1}}(x, L)).
\end{align*}
Continuing with the SAM beamformer, we have by lemma \ref{r_inv_lemma} that
\[
L^\intercal R^{-1} N R^{-1} L = L^\intercal N^{-1} L - \frac{\rho_{x, N}}{\langle N^{-1} x, x\rangle} \left(L^\intercal N^{-1} x \right) \left(L^\intercal N^{-1} x \right)^\intercal
\]
and 
\[
L^\intercal R^{-1} L = L^\intercal N^{-1} L - \frac{\mu_{x, N}}{\langle N^{-1} x, x\rangle} \left(L^\intercal N^{-1} x \right) \left(L^\intercal N^{-1} x \right)^\intercal,
\]
and thus
\[
L^\intercal R^{-1} L = L^\intercal R^{-1} N R^{-1} L + \frac{\rho_{x, N} - \mu_{x, N}}{\langle N^{-1} x, x\rangle} \cdot \left(L^\intercal N^{-1} x \right) \left(L^\intercal N^{-1} x \right)^\intercal.
\]
We can now compute
\begin{align*}
P_{\text{SAM}}(L) &= \Trace \left( \left( L^\intercal R^{-1} L\right) \left(L^\intercal R^{-1} N R^{-1} L \right)^{-1}\right) \\
&= k + \frac{\rho_{x, N} - \mu_{x, N}}{\langle N^{-1} x, x\rangle} \cdot
 \left(L^\intercal N^{-1} x\right)^\intercal  \left(L^\intercal R^{-1} N R^{-1} L \right)^{-1} L^\intercal N^{-1} x.  
\end{align*}
By the Sherman-Morrison formula, we have
\[
 \left(L^\intercal R^{-1} N R^{-1} L \right)^{-1} =
 \left(L^\intercal N^{-1} L\right)^{-1} + \frac{\rho_{x, N}}{\langle N^{-1} x, x\rangle} \cdot
 \frac{ \left(L^\intercal N^{-1} L\right)^{-1} \left(L^\intercal N^{-1} x\right) \left( L^\intercal N^{-1} x\right)^\intercal  \left(L^\intercal N^{-1} L\right)^{-1}}{1 - \rho_{x, N} \cdot \text{GOF}_{N^{-1}}(x, L)},
\]
and thus
\begin{align*}
P_{\text{SAM}}(L) - k &= \left(\rho_{x, N} - \mu_{x, N}\right) \cdot \text{GOF}_{N^{-1}}(x, L)
+ \frac{\rho_{x, N} - \mu_{x, N}}{\langle N^{-1} x, x\rangle} \cdot \frac{\rho_{x, N}}{\langle N^{-1} x, x\rangle}
\cdot \frac{\left(\langle N^{-1} x, x\rangle \cdot \text{GOF}_{N^{-1}}(x, L)\right)^2}{1 - \rho_{x, N} \cdot \text{GOF}_{N^{-1}}(x, L)}\\
&= \left(\rho_{x, N} - \mu_{x, N} \right) \cdot \left(\text{GOF}_{N^{-1}}(x, L) - \frac{\rho_{x, N} \cdot \text{GOF}_{N^{-1}}(x, L)^2}{1 - \rho_{x, N} \cdot \text{GOF}_{N^{-1}}(x, L)} \right)\\
&=  \left(\rho_{x, N} - \mu_{x, N} \right) \cdot \frac{\text{GOF}_{N^{-1}}(x, L)}{1 - \rho_{x, N} \cdot \text{GOF}_{N^{-1}}(x, L)}\\
&= \frac{\rho_{x, N} - \mu_{x, N}}{\rho_{x, N}} \cdot \left( \frac{1}{1 - \rho_{x, N} \cdot \text{GOF}_{N^{-1}}(x, L)} - 1\right)\\
&= \frac{1}{\langle N^{-1} x, x\rangle + 2} \left(  \frac{1}{1 - \rho_{x, N} \cdot \text{GOF}_{N^{-1}}(x, L)} - 1\right)
= g_{x, N}(\text{GOF}_{N^{-1}}(x, L)).
\end{align*}
\end{proof}

Finally, we want to show that the vector version of the neural activity index from \cite{van_veen_lcmv_beamforming} is, in general, biased. 

In \cite{van_veen_lcmv_beamforming}, this value is used as follows. For each position, we look at the corresponding leadfield $L \in \mathbb{R}^{n \times 3}$, whose columns are given by the leadfield vectors corresponding to dipoles at this position with the unit vectors as moments. Then, one computes the value $\widetilde{\text{NAI}}(L)$ as defined in (\ref{nai_van_veen_definition}). The source position is then estimated by searching for large values of $\widetilde{\text{NAI}}(L)$. 

Assume that the signal is originating from position $x_0$ with orientation $\eta_0 \in \mathbb{R}^3$. Denote by $L_0 \in \mathbb{R}^{n \times 3}$ the complete leadfield at this position. By saying that the procedure defined in the last paragraph is \textit{biased}, we mean that, in general, there exists a leadfield $L \in \mathbb{R}^{n \times 3}$ such that
\[
\widetilde{\text{NAI}}(L) > \widetilde{\text{NAI}}(L_0),
\]
i.e. that the map $A \mapsto \widetilde{\text{NAI}}(A)$ do not attain its maximum at the true leadfield.

\begin{theorem}[Biasedness of the original NAI]
\label{original_nai_biased_theorem}
We assume a noisy single source setting, i.e. the data $d \in \mathbb{R}^n$ is given by
\[
d = L_0 \cdot \eta_0 \cdot s + n,
\]
where $s$ is the source amplitude, $\eta_0$ the source orientation, $L_0 \in \mathbb{R}^{n \times 3}$ the leadfield matrix, and $n$ the noise vector. Furthermore, we assume that $L_0$ has full rank, that the source amplitude $s$ is uncorrelated to the components of the noise vector $n$, and that $s$ and $n$ are square integrable. Let $N^{-\frac{1}{2}} L_0 = U S V^\intercal$ be the thin singular value decomposition, and let $S V^\intercal \eta_0 = (v_1, v_2, v_3)$.
Then, if the condition $|v_1| =  |v_2| = |v_3|$ is false, we can find a leadfield $L \in \mathbb{R}^{n \times 3}$ such that
\[
 \widetilde{\text{NAI}}(L) > \widetilde{\text{NAI}}(L_0).
 \]
 
 Since this is almost always the case, we see that $\widetilde{\text{NAI}}$ in general leads to a biased estimation.
\end{theorem}

\begin{proof}
In the proof of theorem \ref{sam_and_nai_dipole_scanning} we computed 
\[
\left( L^\intercal R^{-1} L\right)^{-1} = \left( L^\intercal N^{-1} L\right)^{-1} + \frac{\mu_{x, N}}{\langle N^{-1} x, x\rangle}  
\cdot \frac{\left( L^\intercal N^{-1} L\right)^{-1} L^\intercal N^{-1} x \left(L^\intercal N^{-1} x \right)^\intercal \left( L^\intercal N^{-1} L\right)^{-1}}{1 - \mu_{x, N} \cdot \text{GOF}_{N^{-1}}(x, L) },
\]
where $x = \sqrt{\mathbb{E}\left[s^2\right]} \cdot L_0 \cdot \eta_0$. We thus have for an arbitrary leadfield matrix $L \in \mathbb{R}^{n \times 3}$ that
\begin{align*}
\widetilde{\text{NAI}}(L) &= \frac{\Trace\left( \left(L^\intercal R^{-1} L\right)^{-1}\right)}{\Trace \left(\left(L^\intercal N^{-1} L\right)^{-1} \right)} \\
&= 1 + \frac{\mu_{x, N}}{\langle N^{-1} x, x\rangle} \cdot \frac{1}{1 - \mu_{x, N} \cdot  \text{GOF}_{N^{-1}}(x, L)}
\cdot \frac{\left(L^\intercal N^{-1} x \right)^\intercal \left( L^\intercal N^{-1} L\right)^{-2} L^\intercal N^{-1} x}{ \Trace\left( \left( L^\intercal N^{-1} L\right)^{-1} \right)}.
\end{align*}
Since $N^{-\frac{1}{2}} L_0 = U S V^\intercal$ is a thin singular value decompositions, we have $\left(L^\intercal N^{-1} L\right)^{-1} = V S^{-2} V^\intercal$, $\left(L^\intercal N^{-1} L\right)^{-2} = V S^{-4} V^\intercal$, and $\left(N^{-\frac{1}{2}} L\right) \left( L^\intercal N^{-1} L \right)^{-2} \left( N^{-\frac{1}{2}} L\right)^\intercal = U S^{-2} U^\intercal$. Using this, we get
\[
\widetilde{\text{NAI}}(L_0) = 1 +  \frac{\mu_{x, N}}{\langle N^{-1} x, x\rangle} \cdot \frac{1}{1 - \mu_{x, N} \cdot  \text{GOF}_{N^{-1}}(x, L_0)} \cdot
\frac{\left(U^\intercal N^{-\frac{1}{2}} x \right)^\intercal S^{-2} \left(U^\intercal N^{-\frac{1}{2}} x \right)}{\Trace\left(S^{-2} \right)}.
\]
From lemma \ref{gof_expression},  we have $\text{GOF}_{N^{-1}}(x, L_0) = \frac{\left(N^{-\frac{1}{2}} x \right)^\intercal U U^\intercal \left(N^{-\frac{1}{2}} x \right)}{\langle N^{-1} x, x \rangle}$. Notice that this expression is independent of $S$. In particular, we can change the diagonal values of $S$ without changing the goodness of fit.

Since $L_0$ is full rank, we have $s_1, s_2, s_3 > 0$. Let $0 < \epsilon < \min\{s_1^{-2}, s_2^{-2}, s_3^{-2}\}$. Since, by assumption, the statement $|v_1| = |v_2| = |v_3|$ is false, we have indices $i, j$ with $i \neq j$ and $|v_i| < |v_j|$. Denote by $k$ the remaining index so that $\{i, j, k\} = \{1, 2, 3\}$. Now define $\tilde{s}_i = \left(s_i^{-2} - \epsilon\right)^{-\frac{1}{2}}$, $\tilde{s}_j = \left(s_j^{-2} + \epsilon\right)^{-\frac{1}{2}}$, $\tilde{s}_k = s_k$. Then let $\widetilde{S} = \text{diag}(\tilde{s}_1, \tilde{s}_2, \tilde{s}_3)$. Then, we have $\Trace\left(S^{-2}\right) = \Trace\left( \widetilde{S}^{-2}\right)$, $\tilde{s}_i^{-2} = s_i^{-2} - \epsilon$, and $\tilde{s}_j^{-2} = s_j^{-2} + \epsilon$. Now notice that 
$U^\intercal N^{-\frac{1}{2}} x = \sqrt{\mathbb{E}\left[s^2\right]} S V^\intercal \eta_0 = \sqrt{\mathbb{E}\left[s^2\right]} \cdot (v_1, v_2, v_3)$. Using this, we get
\begin{align*}
\left(U^\intercal N^{-\frac{1}{2}} x \right)^\intercal S^{-2} \left(U^\intercal N^{-\frac{1}{2}} x \right)
= \mathbb{E}\left[s^2\right] \cdot \sum_{l = 1}^3 s_l^{-2} v_l^2
< \mathbb{E}\left[s^2\right] \cdot \sum_{l = 1}^3 \tilde{s}_l^{-2} v_l^2
= \left(U^\intercal N^{-\frac{1}{2}} x \right)^\intercal \widetilde{S}^{-2} \left(U^\intercal N^{-\frac{1}{2}} x \right).
\end{align*}
If we now define $\widetilde{L} = N^{\frac{1}{2}} U \widetilde{S} V^\intercal$, we get
\begin{multline}
\widetilde{\text{NAI}}(L_0) = 1 +  \frac{\mu_{x, N}}{\langle N^{-1} x, x\rangle} \cdot \frac{1}{1 - \mu_{x, N} \cdot  \text{GOF}_{N^{-1}}(x, L_0)} \cdot
\frac{\left(U^\intercal N^{-\frac{1}{2}} x \right)^\intercal S^{-2} \left(U^\intercal N^{-\frac{1}{2}} x \right)}{\Trace\left(S^{-2} \right)} \\
< 1 +  \frac{\mu_{x, N}}{\langle N^{-1} x, x\rangle} \cdot \frac{1}{1 - \mu_{x, N} \cdot  \text{GOF}_{N^{-1}}(x, \widetilde{L})} \cdot
\frac{\left(U^\intercal N^{-\frac{1}{2}} x \right)^\intercal \widetilde{S}^{-2} \left(U^\intercal N^{-\frac{1}{2}} x \right)}{\Trace\left(\widetilde{S}^{-2} \right)}
= \widetilde{\text{NAI}}(\widetilde{L}).
\end{multline}
Thus, we have found a leadfield matrix $\widetilde{L}$ such that $\widetilde{\text{NAI}}(L_0) < \widetilde{\text{NAI}}(\widetilde{L})$. In fact, since the map $L \mapsto \widetilde{\text{NAI}}(L)$ is continuous, we can find a small perturbation $L$ of this $\widetilde{L}$ such that still $\widetilde{\text{NAI}}(L_0) < \widetilde{\text{NAI}}(L)$ and, additionally, $\text{GOF}_{N^{-1}}(x, L) < \text{GOF}_{N^{-1}}(x, L_0)$. Hence, we can even find a leadfield matrix that simultaneously has a higher value for $\widetilde{\text{NAI}}$ and a worse fit to the data than $L_0$.
\end{proof}


\section*{Appendix}

Here, we want to show how to derive equation (\ref{derivative_noise_distortion}).

\begin{lemma}
\label{conjugation_derivative}
Let $h: U \rightarrow \mathbb{R}^{l \times l}$, where $U \subset \mathbb{R}^{k \times l}$ is open, be differentiable. Let
\[
f : U \rightarrow \mathbb{R}^{k \times k}; A \mapsto A h(A) A^\intercal. 
\]
Then $f$ is differentiable with
\[
Df(A)(B) = A h(A) B^\intercal + B h(A) A^\intercal + A Dh(A)(B) A^\intercal.
\]
\end{lemma}

\begin{proof}
Let $B \in B_{\epsilon}(A)$. Then
\begin{multline*}
f(A + B) = (A + B) h(A + B) (A + B)^\intercal
= (A + B) (h(A) + Dh(A)(B) + \varphi_h(B)) (A^\intercal + B^\intercal) \\
= f(A) + A h(A) B^\intercal + B h(A) A^\intercal + A Dh(A)(B) A^\intercal \\
+ A Dh(A)(B) B^\intercal + A \varphi_h(B) A^\intercal + A \varphi_h(B) B^\intercal
+ B h(A) B^\intercal  + B Dh(A)(B) A^\intercal + B Dh(A)(B) B^\intercal + B \varphi_h(B) A^\intercal + B \varphi_h(B) B^\intercal \\
= f(A) + A h(A) B^\intercal + B h(A) A^\intercal + A Dh(A)(B) A^\intercal + o(B).
\end{multline*}
\end{proof}

\begin{lemma}
\label{matrix_inverse_derivative}
Let $U \subset \mathbb{R}^{l \times l}$ be the set of invertible matrices. Let $f : U \rightarrow U; A \mapsto A^{-1}$. Then for $A \in U$ we have
\[
Df(A)(B) = - A^{-1} B A^{-1}.
\]
\end{lemma}

\begin{proof}
Let $A \in \mathbb{R}^{l \times l}$. By the Woodbury identity\footnote{\url{https://en.wikipedia.org/wiki/Woodbury_matrix_identity}} (with $U = B, C = \Id_l, V = \Id_l$), we have
\begin{multline*}
(A + B)^{-1} = A^{-1} - A^{-1} B (\Id + A^{-1} B)^{-1} A^{-1} 
= A^{-1} - A^{-1} B A^{-1} + A^{-1} B \left(\Id - \left(\Id + A^{-1} B \right)^{-1} \right) A^{-1} \\
= A^{-1} - A^{-1} B A^{-1} + o(B).
\end{multline*}
\end{proof}

Now let $f_1(A) = A^\intercal C A$, $f_2(A) = A^{-1}$, and $f_3(A) = A f_2(f_1(A)) A^\intercal$. Then
\[
f_1(A + B) = A^\intercal C A + A^\intercal C B + B^\intercal C A + B^\intercal C B
= f_1(A) + A^\intercal C B + B^\intercal C A + o(B),
\]
and thus
\[
Df_1(A)(B) = A^\intercal C B + B^\intercal C A.
\]
Then, by the chain rule and lemma \ref{matrix_inverse_derivative}, we have
\begin{multline}
D(f_2 \circ f_1)(A)(B) = Df_2(f_1(A)) \circ Df_1(A) (B) = Df_2(A^\intercal C A)(A^\intercal C B + B^\intercal C A) \\
= - (A^\intercal C A)^{-1} (A^\intercal C B + B^\intercal C A) (A^\intercal C A)^{-1}.
\end{multline}

Using lemma \ref{conjugation_derivative}, this now implies
\begin{multline*}
D(f_3)(A)(B) = A f_2 \circ f_1(A) B^\intercal + B f_2 \circ f_1(A) A^\intercal + A D(f_2 \circ f_1)(A)(B) A^\intercal \\
= A (A^\intercal C A)^{-1} B^\intercal + B (A^\intercal C A)^{-1} A^\intercal 
- A (A^\intercal C A)^{-1} (A^\intercal C B + B^\intercal C A) (A^\intercal C A)^{-1} A^\intercal.
\end{multline*}

Now, note that for a linear map $T$ we have $DT(A) = T$. Letting $g(A) = \Trace \left( \left(A^\intercal C A \right)^{-1} A^\intercal C N C A\right)$, the equation $\Trace(E \cdot F) = \Trace(F \cdot E)$ implies
\[
g(A) = \Trace \left( N^{\frac{1}{2}} C A \left( A^\intercal C A\right)^{-1} A^\intercal C N^{\frac{1}{2}}\right)
= \Trace\left(N^{\frac{1}{2}} C f_3(A)  C N^{\frac{1}{2}}\right),
\]
and since the map $D \mapsto \Trace\left(N^{\frac{1}{2}} C D C N^{\frac{1}{2}} \right)$ is linear, the chain rule implies
\begin{multline*}
Dg(A)(B) = \Trace\left(N^{\frac{1}{2}} C \left(D f_3(A)(B) \right)C N^{\frac{1}{2}} \right) \\
= \Trace\left(N^{\frac{1}{2}} C ( A (A^\intercal C A)^{-1} B^\intercal + B (A^\intercal C A)^{-1} A^\intercal 
- A (A^\intercal C A)^{-1} (A^\intercal C B + B^\intercal C A) (A^\intercal C A)^{-1} A^\intercal) C N^{\frac{1}{2}} \right) \\
=\Trace\left(N^{\frac{1}{2}} C \left(B \left(A^\intercal C A\right)^{-1} A^\intercal - A \left(A^\intercal C A\right)^{-1} A^\intercal C B\left(A^\intercal C A\right)^{-1} A^\intercal\right) C N^{\frac{1}{2}} \right) \\
+ \Trace\left( N^{\frac{1}{2}} C \left( A \left(A^\intercal C A\right)^{-1} B^\intercal - A \left(A^\intercal C A\right)^{-1} B^\intercal C A\left(A^\intercal C A\right)^{-1} A^\intercal\right) C N^{\frac{1}{2}}\right) \\
\stackrel{\text{$\Trace(F^\intercal) = \Trace(F)$}}{=}
2 \cdot \Trace\left(N^{\frac{1}{2}} C \left(B \left(A^\intercal C A\right)^{-1} A^\intercal - A \left(A^\intercal C A\right)^{-1} A^\intercal C B\left(A^\intercal C A\right)^{-1} A^\intercal\right) C N^{\frac{1}{2}} \right) \\
= 2 \cdot \Trace\left(B \left(A^\intercal C A\right)^{-1} A^\intercal  C N C\right)
- 2 \cdot \Trace\left(B  \left(A^\intercal C A\right)^{-1} A^\intercal C N C  A \left(A^\intercal C A\right)^{-1} A^\intercal C\right)\\
= 2 \cdot \Trace\left(B \cdot \left( \left(A^\intercal C A\right)^{-1} A^\intercal  C N C - \left(A^\intercal C A\right)^{-1} A^\intercal C N C  A \left(A^\intercal C A\right)^{-1} A^\intercal C\right) \right).
\end{multline*}

This is now exactly equation (\ref{derivative_noise_distortion}).


\bibliographystyle{ieeetr}
\bibliography{sources}

\end{document}